\documentclass[journal]{IEEEtran}

\usepackage[hyphens]{url}
\usepackage{graphicx}
\usepackage[numbers,sort]{natbib}

\usepackage{amsmath}
\usepackage{amssymb}
\usepackage{amsthm}
\usepackage{booktabs}
\usepackage[table]{xcolor}
\usepackage{pifont}
\newcommand{\yes}{\textcolor{green!55!black}{\ding{51}}}
\newcommand{\no}{\textcolor{red!75!black}{\ding{55}}}
\definecolor{zebra}{gray}{0.93}
\usepackage{placeins}
\usepackage{tcolorbox}
\tcbuselibrary{breakable,skins}

\newtheorem{definition}{Definition}

\newtheorem{proposition}{Proposition}

\newtheorem{lemma}{Lemma}
\theoremstyle{definition}
\newtheorem*{vulnerability}{Vulnerability}

\theoremstyle{remark}

\tcbset{paperbox/.style={enhanced,unbreakable,
  colback=white,colframe=black!75,boxrule=0.8pt,arc=1pt,
  left=4pt,right=4pt,top=3pt,bottom=3pt}}
\newtcolorbox{takeaway}{paperbox,colframe=green!45!black,colback=green!5,
  before upper={\clubpenalty10000\widowpenalty10000{\normalfont\bfseries Takeaway.}\ \itshape}}
\newtcolorbox{keyfinding}{paperbox,colframe=blue!60!black,colback=blue!4,
  before upper={\clubpenalty10000\widowpenalty10000{\normalfont\bfseries Key Finding.}\ \itshape}}
\newtcolorbox{researchgap}{paperbox,colframe=red!55!black,colback=red!4,
  before upper={\clubpenalty10000\widowpenalty10000{\normalfont\bfseries Research Gap.}\ \itshape}}

\title{LLM Misuse Safeguards Rely on Untrusted Input:\\
A Widespread, High-Risk Vulnerability}

\author{Pingyu~Wu\textsuperscript{1,2}, Weiming~Zhang\textsuperscript{1,*}, and Nenghai~Yu\textsuperscript{1}\\[6pt]
{\normalsize \textsuperscript{1}University of Science and Technology of China, Hefei, China\\
\textsuperscript{2}Hefei AiDA Lab, Hefei, China\\
wupingyu@mail.ustc.edu.cn, zhangwm@ustc.edu.cn\\
\textsuperscript{*}Corresponding author}%
\thanks{This work has been submitted to the IEEE for possible publication. Copyright may be transferred without notice, after which this version may no longer be accessible.}}

\begin{document}
\bstctlcite{IEEEexample:BSTcontrol}

\maketitle

\begin{abstract}
Large language models now help professionals in fields such as cybersecurity and biology, yet the same answers would help an attacker. Deployers rely on misuse safeguards to keep such answers from attackers. We find that these safeguards share a high-risk vulnerability: they decide whether to release an answer from what users write, including the professional identity or purpose they claim, and anyone can write such a claim. Threat reports from model vendors and security firms show that attackers who fabricate a professional identity or purpose already exploit this vulnerability to cause serious harm. The vulnerability is also widespread across the state of the art: every defense proposed in 129 papers from the latest editions of 12 top venues has it. Safeguard research studies the wrong object: it treats harm as a property of the text of a request, although the same text can come from a professional or an attacker. Current safeguard benchmarks therefore cannot reveal the vulnerability: a safeguard can score high on them and still give dangerous answers to attackers. Because the vulnerability lies in what users write, no improvement in how safeguards judge that writing can close it. Under the patch we propose, large language models can help professionals without giving dangerous capabilities to anyone who merely claims to be one.

\end{abstract}

\begin{IEEEkeywords}
Access control, dual use, guardrails, large language models, misuse safeguards, security evaluation, untrusted input
\end{IEEEkeywords}

\section{Introduction}\label{sec:intro}

Large language models (LLMs) now help professionals in fields such as
cybersecurity and biology~\cite{anthropic2026s1,anthropic2026lsvp}, yet the
same answers would help an attacker. Anthropic reports that independent
researchers have used its models to identify and exploit vulnerabilities in
other companies' systems, some within authorized bug bounty programs, and
warns that ``comparable capabilities could be used by malicious actors without
authorization''~\cite[p.~45]{anthropic2026s1}. In biology, the UK AI Security
Institute reports that non-experts who used frontier models to write protocols
for viral recovery had 4.7 times the odds of writing a feasible protocol
compared with those using the internet alone~\cite{aisi2025frontiertrends}.
Such content is dual-use: the same work can serve legitimate and malicious
purposes~\cite{nrc2004biotechnology,miller2007dualuse}. Deployers rely on
misuse safeguards to keep it from attackers.

\begin{definition}[Goal of a safeguard]\label{def:goal}
A safeguard aims to block requests made with malicious intent
without affecting legitimate users' use of the model~\cite{sharma2025constitutional}.
\end{definition}

To meet this goal, vendors and researchers have built safeguards of many
kinds: safety training and reasoning over safety
policies~\cite{guan2024deliberative}, classifiers on requests and
answers~\cite{inan2023llamaguard,han2024wildguard}, Constitutional
Classifiers~\cite{sharma2025constitutional,cunningham2026constitutional},
and classifiers that infer the user's
intent~\cite{ferrao2026trueintents}. These safeguards block requests that
plainly reveal malice: when external red teamers tested Anthropic's model
Fable~5.1, requests that stated the overall objective were
blocked~\cite[p.~57]{anthropic2026fable51card}. On a dual-use request that
does not plainly reveal malice, however, a safeguard must still judge how the
answer will be used, and it judges from the text of the
request~\cite{han2024wildguard,sharma2025constitutional}. All the text
conveys about the use is what the user claims, such as a professional identity or purpose, and anyone can write such a claim.

\begin{vulnerability}
Current misuse safeguards base their security decision on untrusted input.
\end{vulnerability}
MITRE already lists this weakness as CWE-807, reliance on untrusted inputs in a security decision~\cite{mitre2026cwe807}.

Attackers already exploit the vulnerability to cause serious harm: a
written claim has sufficed to obtain the models' help
(Section~\ref{sec:cases}). A criminal group used commercial models in
attacks on more than 400 Latin American financial targets, bypassing the
models' safety policies by describing its work as an authorized red-team
engagement or academic penetration
testing~\cite[pp.~3--7, 11]{straiker2026carliy0s}. The campaign ran on
Claude Opus~4.7 and later
models~\cite[p.~3]{straiker2026carliy0s}, despite the safeguards that
Anthropic deployed to block high-risk cybersecurity
uses~\cite{anthropic2026opus47}. Google and Anthropic report the same
method in further attacks, including malware development and espionage
against government
agencies~\cite{google2025gtigtracker,anthropic2025espionage}.

The vulnerability is also widespread. For unverified users, Anthropic and OpenAI leave high-risk dual-use content to safeguards that decide from what users write: Anthropic in cybersecurity and biology, OpenAI in cybersecurity~\cite{anthropic2026safeguardsframework,openai2026trustedcyber,anthropic2025asl3} (Section~\ref{sec:default-access}). Safeguard research builds its defenses the same way and studies the wrong object: it treats harm as a property of the text of a request.

Wybitul observes that current safety systems infer context from user-supplied input that attackers can fabricate and uses this observation to motivate access control~\cite{wybitul2025access}. This paper establishes the weakness as a vulnerability with the four properties below
and proposes a patch under which LLMs can help professionals who hold a credential without giving
dangerous capabilities to anyone who merely claims to be one
(Section~\ref{sec:credentials}).
\begin{itemize}
\item \textbf{High-risk vulnerability.} Attackers can obtain whatever
dual-use content these safeguards release to professionals
(Section~\ref{sec:obtain}).
\item \textbf{Widespread in state-of-the-art defenses.} Every defense
proposed in 129 safeguard papers from the latest editions of 12 top venues
has the vulnerability (Section~\ref{sec:census}).
\item \textbf{Hidden from current benchmarks.} A safeguard can score high
on current benchmarks and still release dangerous dual-use content to
attackers (Sections~\ref{sec:evaluations} and~\ref{sec:scoring}).
\item \textbf{Persistent in better safeguards.} No improvement in how a
safeguard judges what users write can close the vulnerability
(Section~\ref{sec:limits}).
\end{itemize}

The rest of this paper is organized as follows. Section~\ref{sec:background}
describes what a safeguard must judge and how deployed safeguards decide.
Section~\ref{sec:evaluations} shows that a label on the text of a dual-use request that does not plainly reveal malice is wrong for some of the users who send it. Section~\ref{sec:obtain} proves that attackers can obtain whatever
dual-use content these safeguards release to professionals, and
Section~\ref{sec:cases} documents attacks that exploit the vulnerability in
deployed systems. Section~\ref{sec:census} measures its prevalence in
state-of-the-art defenses, and Section~\ref{sec:limits} shows why better
safeguards cannot close it. Section~\ref{sec:credentials} presents the
patch, in which safeguards block requests that plainly reveal malice,
content-category classifiers apply each category's release rule, and
credentials decide who receives dangerous dual-use content.
Section~\ref{sec:related} discusses related work, and
Section~\ref{sec:conclusion} concludes.

\section{What a Safeguard Must Judge}
\label{sec:background}

A misuse safeguard aims to block requests made with malicious intent without affecting legitimate users (Definition~\ref{def:goal}). Deployments also run other mechanisms that decide what a model releases. A content-category classifier judges what the content of an answer is; a safeguard judges how the answer will be used (Table~\ref{tab:mechanisms}). Vendors leave the judgment on high-risk dual-use content to safeguards (Section~\ref{sec:default-access}), and the threat these safeguards face is a user who puts the answer to harmful use (Section~\ref{sec:threat-model}).

\subsection{Safeguards and Content-Category Classifiers}
\label{sec:mechanisms}

Standards bodies and vendors sort what a model can produce by type, and this sorting shows which content requires a judgment of use. For chemical and biological capabilities, NIST AI~800-1 sets apart ``Predominantly Harmful capabilities with very limited legitimate applications'' from dual-use capabilities, which have legitimate applications~\cite{nist2025ai8001}. Anthropic's framework for its cyber safeguards places prohibited uses in a tier of their own, separate from its high-risk and low-risk dual-use tiers and from benign uses~\cite{anthropic2026safeguardsframework}. For prohibited and benign content, the type settles the release: a deployment withholds the first from everyone and releases the second to everyone.

\begin{table}[t]
\centering
\caption{Mechanisms that decide what a model releases.}
\label{tab:safeguard-vs-category}\label{tab:mechanisms}
\footnotesize
\setlength{\tabcolsep}{3pt}
\renewcommand{\arraystretch}{1.1}
\rowcolors{2}{white}{zebra}
\begin{tabular}{p{\dimexpr0.26\linewidth-5pt\relax}p{\dimexpr0.14\linewidth-1pt\relax}p{0.22\linewidth}cc}
\toprule
Mechanism & Judges & Reads & \shortstack{Input\\settles\\judgment} & \shortstack{Decision\\unaffected\\by claims}\tabularnewline
\midrule
\raggedright Content-category classifier~\cite{openai2025gpt5card} & \raggedright Content & \raggedright Context, answer & \yes & \yes\tabularnewline
\raggedright Safeguard~\cite{sharma2025constitutional,ferrao2026trueintents} & \raggedright Use & \raggedright Context, activations, answer & \no$^{*}$ & \no\tabularnewline
\raggedright Credential~\cite{anthropic2026cvp,anthropic2026lsvp} & \raggedright Credential & \raggedright Credential & \yes & \yes\tabularnewline
\bottomrule
\hiderowcolors
\addlinespace[2pt]
\multicolumn{5}{l}{$^{*}$\yes{} when the text plainly reveals malice.}\tabularnewline
\end{tabular}
\end{table}

\emph{Content-category classifiers} make the release decisions that the type settles (Table~\ref{tab:safeguard-vs-category}). They classify the answer's content and apply the release rule for that category. Withholding content from everyone, without judging who asks, is \emph{capability reduction}; unlearning and safe-completion training also reduce capability~\cite{li2024wmdp,yuan2025safecompletions}.

On dual-use content, the type does not settle the release. A deployer that treats every user alike must either withhold the content---and professionals lose it---or release it---and attackers receive it. Dual-use content is therefore the content that a safeguard exists for: to meet its goal (Definition~\ref{def:goal}), the safeguard must judge, request by request, how the answer will be used.

The text of a request shows which deliverable the user asks for, but for dual-use work it often does not show how that deliverable will be used. Vendors that operate such safeguards say the same. Anthropic's Life Sciences Verification Program (LSVP) states that in biology ``it's often not possible to differentiate between a user doing valid work [\ldots] and pursuing harm''~\cite{anthropic2026lsvp}, and OpenAI notes that ``it can be difficult to tell whether any particular cyber action is intended for defensive usage, or to cause harm''~\cite{openai2026trustedcyber}. In Anthropic's account of its cyber safeguards, what separates the legitimate case from the harmful one is ``who is doing the work, and under what authorization''~\cite{anthropic2026safeguardsframework}. Its misuse-detection account gives an example: ``A request for exploit code might come from a security team testing its own software, or from an attacker who is roleplaying as a pentester''~\cite{anthropic2026selfassessment}.

The text does establish harmful use when it plainly reveals malice.
\begin{definition}[Plainly revealing malice]\label{def:plain-malice}
A message \emph{plainly reveals malice} if its text shows that the model's answer will be put to harmful use, such as a stated malicious goal, or shows an attempt to evade the safeguard, such as an instruction to ignore safety rules.
\end{definition}
A safeguard can block such a message on its text alone: when external red teamers tested Fable~5.1, requests that stated the overall objective were blocked~\cite[p.~57]{anthropic2026fable51card}. On the remaining dual-use requests, all the text conveys about the use is what the user claims, such as an identity or purpose.

Content-category classifiers also read context, but a claimed identity or purpose does not change what the content is, so it changes neither the category nor the applicable release rule. A safeguard's decision on the same content, by contrast, can change with what the user writes. Table~\ref{tab:safeguard-vs-category} lists a third mechanism: release by a credential that the user holds (Section~\ref{sec:credentials}). Of the three, only a safeguard is designed to change its decision with what the user claims.

\subsection{Deployed Safeguards}
\label{sec:default-access}

Vendors call mechanisms of both kinds classifiers. To tell which kind a deployed mechanism is, we check whether its decision on the same content can differ from request to request. OpenAI's biological classifiers for GPT-5 are content-category classifiers: they determine which part of a biological threat taxonomy a generated response falls into~\cite[p.~47]{openai2025gpt5card}. The cyber classifier of Anthropic's model Fable~5.1, modeled on Constitutional Classifiers~\cite{cunningham2026constitutional}, decides whether a conversation should be blocked~\cite[pp.~45--46]{anthropic2026fable51card}, and Anthropic states that in both dual-use tiers context can change what is blocked~\cite{anthropic2026safeguardsframework}. This classifier's decision on the same content can therefore differ from request to request, so it is a safeguard. Anthropic relies on it to block the high-risk dual-use tier at general access, for users whom it has not verified~\cite{anthropic2026safeguardsframework}, \cite[p.~52]{anthropic2026fable51card}. What these users write, such as a professional purpose, can therefore decide whether they receive content in that tier.

Other deployments leave high-risk dual-use content to safeguards in the same way. Opus~4.7's safeguards ``detect and block requests that indicate prohibited or high-risk cybersecurity uses''~\cite{anthropic2026opus47}; they judge the use that a request indicates. OpenAI's cyber mitigations for users outside Trusted Access for Cyber train the model ``to refuse clearly malicious requests'' and add classifiers that ``detect potential signals of suspicious cyber activity''~\cite{openai2026trustedcyber}. For biology, Anthropic's AI Safety Level~3 (ASL-3) protections for Claude Opus~4 include Constitutional Classifiers and target misuse for chemical, biological, radiological, and nuclear weapons~\cite{anthropic2025asl3}; these classifiers judge from context whether an exchange is harmful~\cite{sharma2025constitutional}.

The assistance that these safeguards govern provides large uplift in biology: the UK AI Security Institute reports that non-experts who used frontier models to write protocols for viral recovery had 4.7 times the odds (confidence interval 2.8--7.9) of writing a feasible protocol compared with those using the internet alone~\cite{aisi2025frontiertrends}. In each of these deployments, what a user writes can decide whether it receives high-risk dual-use assistance in cybersecurity or biology.

\subsection{Threat Model}\label{sec:threat-model}

The safeguards of Section~\ref{sec:default-access} exist to stop one threat: a user who asks for a deliverable and puts the answer to harmful use. Each \emph{task} names a deliverable that a user asks the model for, and its legitimate users are the professionals who need that deliverable for their work. Users who hold the weights and can change the model are outside this threat model. Because we study harm that arises from how an answer is used, we exclude outputs whose production is itself the harm, such as child sexual abuse material.

\paragraph{Attacker}
The attacker is a user who asks for a task's deliverable and puts the answer to harmful use. Facts such as authorization, ownership, or downstream use separate the attacker from the task's professionals. These facts do not change the task: an attacker who targets an organization's systems asks for the same deliverable as the testers whom that organization authorizes. The request does not establish these facts (Section~\ref{sec:mechanisms}).

\paragraph{Attacker's goal}
The attacker wants the answers that professionals receive for the task. The goal names no particular harm, because an attacker chooses how to use an answer after receiving it.

\paragraph{Capability}
The attacker acts on the deployment only by sending messages. It chooses every message, and a deployment accepts the same messages from every user. The attacker can therefore send any message that a professional can send to receive an answer: a request for the deliverable, an identity or a purpose, and replies to the safeguard's questions. It has no need to plainly reveal malice (Definition~\ref{def:plain-malice}), because the deliverable itself serves the attacker~\cite{nrc2004biotechnology,miller2007dualuse}, \cite[p.~45]{anthropic2026s1}.

\paragraph{Defender}
The deployment is the model together with everything that decides what the model releases. It may read the whole context and any credential and choose its release by any rule, with no computational restriction: it may ask follow-up questions, randomize, stack classifiers, read candidate answers before release, and keep internal state.

\section{Dual-Use Requests Have No Correct Text Label}
\label{sec:evaluations}

A safeguard must judge how the answer to a dual-use request will be used, yet the text of the request settles this only when it plainly reveals malice (Section~\ref{sec:background}). Safeguards learn this judgment from labeled data and are scored against labeled data, both of which attach each safety label to the text of a request, an answer, or a conversation. We show that no text label is correct for a dual-use request that does not plainly reveal malice.

Safeguard research states that a safeguard aims to block requests made with malicious intent without affecting legitimate users (Definition~\ref{def:goal}), so a safety label on an occurrence of a request is correct if it is \emph{harmful} exactly when the sender has malicious intent. The text of such a request can come from a professional or an attacker, so any label that a dataset assigns is wrong for one of them. If attackers send a fraction \(\pi\) of its occurrences, with \(0<\pi<1\), the label \emph{harmful} is wrong on a fraction \(1-\pi\) and the label \emph{benign} is wrong on a fraction \(\pi\):
\begin{equation}
\text{label error}\;\geq\;\min\{\pi,\,1-\pi\}\;>\;0 .
\label{eq:label-error}
\end{equation}

A label attached to the text can depend only on how the text looks, so safeguards trained to reproduce such labels learn whether a text looks malicious, not how its answer will be used. Benchmarks do not expose this gap: several leave these requests out of what they score, and those that retain them label each one by its text. Accuracy, attack success rates, and refusal rates computed against such labels therefore cannot show that a safeguard distinguishes malicious use from legitimate use on these requests.

\subsection{Training on Text Labels}
\label{sec:training}

The training data behind Llama Guard~\cite{inan2023llamaguard},
WildGuard~\cite{han2024wildguard}, BeaverTails~\cite{ji2023beavertails},
Aegis~\cite{ghosh2024aegis,ghosh2025aegis2},
ToxicChat~\cite{lin2023toxicchat}, ShieldGemma~\cite{zeng2024shieldgemma},
Constitutional Classifiers~\cite{sharma2025constitutional}, and OpenAI's
moderation classifier~\cite{markov2023holistic} all label text under a
written definition of harmful content. OpenAI notes that undesired content varies with use
cases and users~\cite{markov2023holistic}, and BeaverTails observes that a
potentially harmful inquiry need not be malevolent~\cite{ji2023beavertails}; both still label the text.

Two datasets treat a request as safe when its text does not settle safety. Aegis~2.0 offers a Needs Caution label for ``ambiguous
cases where there is not enough context to determine safety'' and maps it
to Safe in every reported experiment~\cite{ghosh2025aegis2}.
Constitutional Classifiers label a request harmless when it has no
obviously suspicious features, although it could be harmful in other
circumstances~\cite{sharma2025constitutional}. The safety training dataset WildJailbreak likewise removes
rewrites that ``overly conceal the harmful intent of the original vanilla
prompt''~\cite{jiang2024wildteaming}, so text that conceals harmful intent too well does not appear among its harmful examples.

On dual-use requests, these labels record only whether a user's text looks malicious. Two studies find that the detectors they examine rely on surface cues of this kind: probing-based detectors rely on instructional patterns and trigger
words~\cite{wang2025falseconfidence}, and aligned models rely on surface
semantic and stylistic cues~\cite{wu2025outcomeaware}. These cues lie in the wording, which the user chooses (Section~\ref{sec:threat-model}), so an attacker can request the same work without them.

\subsection{Excluded Dual-Use Requests}
\label{sec:benchmarks}

Dual-use requests that do not plainly reveal malice would expose a safeguard that has learned only these cues, yet several benchmarks drop such requests or set no correct response for them.
Do-Not-Answer rewrites borderline questions to specify a risky scenario or
drops them~\cite{wang2023donotanswer}. HarmBench removes or alters candidate behaviors so that ``it would be highly unlikely for most people to want an LLM to perform the behavior for benign reasons''~\cite{mazeika2024harmbench}. PHTest labels separately the prompts whose harmfulness depends on intent that ``can be impossible to
discern'' and treats refusing them as a developer preference~\cite{an2024phtest}.
None of these benchmarks test safeguards on the requests that
safeguards exist to handle.

\subsection{One Label per Text}
\label{sec:dual-use}

Benchmarks that keep dual-use requests assign each one a single label---the same whoever sends it---choosing the label in one of three ways.

\emph{Label the request legitimate from its stated purpose.}
CarryOnBench counts it as correct when benign users who clarify their purpose recover help previously withheld~\cite{zheng2026carryon}. A safeguard that scores well on it therefore also helps attackers who make the same declaration, as the criminal group did when it described its work as an authorized red-team engagement (Section~\ref{sec:premises}).

\emph{Label the request benign and count a block as an error.}
CyberSecEval~2 counts refusal of its benign-labeled dual-use prompts as a
false refusal, although its authors note that many test prompts ``could
be equally construed as safe or unsafe''~\cite{bhatt2024cyberseceval2}. OR-Bench labels its borderline prompts safe and counts
refusals of them as
over-refusal~\cite{cui2025orbench}; JailbreakBench's
benign behaviors also include borderline cases~\cite{chao2024jbbempirical}.

\emph{Label the request harmful and count a block as correct.}
SORRY-Bench's
additional unsafe instructions include Misrepresentation variants, which
add an account of the user or circumstances that the benchmark treats as false~\cite{xie2025sorrybench,zeng2024johnny},
and it scores a lower fulfillment rate on them as safer.

Legitimate and benign labels reward releasing the answer to the attackers who send the request; harmful labels reward withholding it from the professionals who send it (Equation~\eqref{eq:label-error}). On requests that state a purpose, CarryOnBench counts help as correct and SORRY-Bench counts refusal as correct; each is right only for some senders, because the stated purpose is true when a professional sends the request and false when an attacker does.

\subsection{What High Scores Show}
\label{sec:scores}

A high score demonstrates two abilities, both on decisions that the text alone settles for every user. On requests that plainly reveal malice, it shows that the safeguard blocks them. On content-category labels, it shows that the classifier sorts content correctly. On the remaining dual-use requests---which attackers send (Section~\ref{sec:cases}) and benchmarks drop or label by their text---it does not show whether the safeguard tells attackers from professionals.

The text of such a request is the same whoever sends it, so a safeguard that reads only the request receives the same input from every sender. A benchmark that labeled each occurrence by the user who sent it would label every occurrence correctly, but the safeguard would still score no higher than a rule that always predicts the more common user type.

\begin{proposition}[Scores on a shared request]
\label{cor:scores}
Let \(x\) be a dual-use request that does not plainly reveal malice and that both attackers and professionals send, let a safeguard \(D\) receive the same input from every user who sends \(x\), and let \(p\) be the probability that \(D\) refuses \(x\).
\begin{enumerate}
\item The attack success rate \(1-p\) measured on \(x\) is also the rate at which the professionals who send \(x\) receive the answer.
\item If attackers send a fraction \(\pi\) of the occurrences of \(x\) and exactly these occurrences are labeled harmful, the expected accuracy of \(D\) on \(x\) is \(\pi p+(1-\pi)(1-p)\leq\max\{\pi,1-\pi\}\), the accuracy of a rule that ignores the request and always predicts the more common user type.
\end{enumerate}
\end{proposition}

\begin{proof}
\(D\) receives the same input from every user who sends \(x\), so it refuses each with probability \(p\), giving item~1. In item~2, \(D\) is correct with probability \(p\) on an attacker's occurrence and \(1-p\) on a professional's, so the accuracy, a weighted mean of \(\pi\) and \(1-\pi\), is at most their maximum.
\end{proof}

An attack success rate on such a request therefore shows how often the safeguard answers everyone who sends it; a lower rate shows only that it withholds the answer more often from everyone. By item~2, the benchmark accuracy of a safeguard that receives the same input from every sender can exceed that of always predicting the more common user type only because the benchmark gives every sender of a request the same label (Section~\ref{sec:dual-use}) or because the requests labeled harmful differ in text from those labeled benign---for example by plainly revealing malice or by carrying the surface cues of Section~\ref{sec:training}. Neither shows that a safeguard distinguishes attackers from professionals who send the same messages. A safeguard can therefore score high by blocking requests whose text plainly reveals malice or carries these cues and answering the rest; by item~1 of Proposition~\ref{cor:scores}, it answers the attackers who send the rest as often as the professionals.

\begin{keyfinding}
A safeguard can score high on benchmarks and still release dangerous dual-use content to attackers.
\end{keyfinding}

\section{Safeguards Decide from Untrusted Input}
\label{sec:obtain}\label{sec:context-def}

On a dual-use request that does not plainly reveal malice, the text does not settle how the answer will be used, so no text label of the request is correct (Section~\ref{sec:evaluations}). Deployed safeguards still decide from this text (Section~\ref{sec:default-access}), and all it conveys about the use is the identity or purpose that the user claims. Any user can write any identity and purpose in any wording, and a deployment cannot verify such claims from the request itself: NIST's identity guidelines define a claimed identity as ``an applicant's declaration of unvalidated and unverified personal attributes''~\cite[App.~D]{nist2025sp80063a4}. Deployed safeguards therefore share a vulnerability: they base a security decision on untrusted input, the weakness that MITRE lists as CWE-807~\cite{mitre2026cwe807}.

The untrusted input extends beyond written claims. We call \emph{context} every input determined by what the user sends. The attacker chooses every message it sends and acts on the deployment only by sending messages (Section~\ref{sec:threat-model}), so all context is untrusted input and every input the attacker controls is context. Activations, classifier flags, and other computations on context add no input beyond context, so a safeguard that reads them still decides from untrusted input. Any input that attackers cannot change by what they send is a \emph{credential} (Definition~\ref{def:credential}), such as a vendor's record that an account passed verification.

\begin{definition}[Safeguard that decides only from context]
\label{def:open-access}
A safeguard \emph{decides only from context} if a user affects it only through the messages that the user sends: what it computes, replies, and answers depends on these messages and on its own randomness, and on no credential.
\end{definition}

Consider such a safeguard and a dual-use request that does not plainly reveal malice. By Definition~\ref{def:open-access}, a professional and an attacker who send the same messages give the safeguard the same input, so their answers have the same distribution: if the safeguard refuses with probability~\(p\),
\begin{equation}
\begin{aligned}
&\Pr[\text{an attacker receives the answer}]=1-p\\
&\quad=\Pr[\text{a professional receives the answer}].
\end{aligned}
\label{eq:single}
\end{equation}
Equation~\eqref{eq:single} is the CWE-807 bypass in a misuse safeguard: an attacker who sends a professional's messages receives the answer as often as the professional does. It holds over a whole conversation and whatever the safeguard computes, because the safeguard computes from the same input for both users (Section~\ref{sec:decide-only}). Since every user chooses messages from the same set (Section~\ref{sec:threat-model}), an attacker can obtain any answer the safeguard gives anyone. A stricter rule refuses more often and therefore keeps the answer from attackers only by withholding it from professionals just as often.

For users whom Anthropic has not verified, its safeguard for the high-risk dual-use tier decides only from what they write (Section~\ref{sec:default-access}), so Equation~\eqref{eq:single} applies: whatever high-risk dual-use content it releases to professionals, attackers can obtain as well. Relaxing a safeguard for a stated purpose, such as an authorized red-team engagement, relaxes it for everyone who states that purpose.

\section{Exploitation in Deployed Systems}
\label{sec:cases}

Deployed safeguards decide the release of high-risk dual-use content from context (Section~\ref{sec:default-access}), so an attacker who writes a professional identity or purpose into the context can obtain what professionals obtain (Section~\ref{sec:obtain}). Threat reports confirm that attackers have done so against deployed commercial models in the three cases of Table~\ref{tab:cases}: attacks on banks, malware development, and espionage against government agencies.

\begin{table}[t]
\centering
\caption{Documented misuse of deployed models and the professional identity or purpose that the attackers claimed.}
\label{tab:cases}
\footnotesize
\setlength{\tabcolsep}{3pt}
\renewcommand{\arraystretch}{1.1}
\rowcolors{2}{white}{zebra}
\begin{tabular}{p{\dimexpr0.24\linewidth-6pt\relax}p{\dimexpr0.35\linewidth-6pt\relax}p{\dimexpr0.41\linewidth-6pt\relax}}
\toprule
\raggedright Case (source) & \raggedright Harm & \raggedright Claimed identity or purpose\tabularnewline
\midrule
\raggedright CARLIY0S (Straiker) & \raggedright Attacks on 400+ financial targets & \raggedright Authorized red-team engagement\tabularnewline
\raggedright Several actors (Google) & \raggedright Exploitation and malware tooling & \raggedright Capture-the-flag participant, student, security researcher\tabularnewline
\raggedright GTG-1002 (Anthropic) & \raggedright Espionage targeting about 30 entities & \raggedright Security-firm employees running defensive tests\tabularnewline
\bottomrule
\end{tabular}
\end{table}

\subsection{The CARLIY0S Campaign}
\label{sec:premises}\label{sec:deployed-bypass}

CyberCartel, a cybercriminal group tracked since 2012, used its CARLIY0S system and multiple commercial models to attack more than 400 Latin American financial targets within months~\cite[pp.~3, 6--7]{straiker2026carliy0s}. The report describes the victims as financial institutions, banks, and fintech payment gateways. The harvested material centered on secrets for settlement through SPEI, Mexico's interbank payment system, and included 903 Mexican bank account numbers. Straiker reports that the access this material provides ``could enable abuse of SPEI''~\cite[pp.~4--5, 7]{straiker2026carliy0s}.

The group obtained the models' help with these attacks by bypassing their safeguards with a claimed purpose. Straiker examined the group's own scripts, model configurations, prompts, and rule files~\cite[pp.~3, 7, 11]{straiker2026carliy0s}. Straiker reports that explicit requests for exploitation logic or vulnerability scans would trigger safety refusals; the operators instead described their working environment as an authorized red-team engagement, academic penetration testing, or a vulnerability disclosure program. Straiker calls this ``dual-use contextual framing'' and reports ``a false context that reliably got models to bypass their own safety policies''~\cite[pp.~4, 11]{straiker2026carliy0s}. The operators wrote this description into the agent's rule files and reused it across their prompts.

The safeguards that the CARLIY0S operators bypassed were built to block high-risk cybersecurity uses. The campaign ran on Claude Opus~4.7, 4.8, and 5~\cite[p.~3]{straiker2026carliy0s}, and Opus~4.7's safeguards ``detect and block requests that indicate prohibited or high-risk cybersecurity uses''~\cite{anthropic2026opus47}. These safeguards judge the use that a request indicates. The operators' written purpose made their requests indicate an authorized engagement, so the safeguards did not stop the campaign~\cite[p.~11]{straiker2026carliy0s}. The operators lost access to Claude only after these attacks, when Straiker shared its findings with Anthropic, which then revoked their access~\cite[pp.~5, 23]{straiker2026carliy0s}.

\subsection{The Same Method in Other Attacks}

The method that bypassed Opus~4.7's safeguards recurs across the vendors' own reports. Google reports threat actors posing as capture-the-flag participants, students, or cybersecurity researchers ``to persuade Gemini to provide information that would otherwise be blocked, enabling tool development''. One actor was initially refused, then presented itself as a capture-the-flag participant and received ``helpful information that could be misused to exploit the system''; Google notes that the same prompts from a real participant ``would be benign inquiries''~\cite{google2025gtigtracker}. These prompts therefore did not plainly reveal malice; the help followed from the identity that the actor wrote.

Google describes such identities as routine: threat actors ``often leverage expert cybersecurity personas as a structured approach to prompt Gemini''~\cite{google2026gtigmay}. While developing custom malware, the Iranian group TEMP.Zagros encountered safety responses and used plausible pretexts ``to bypass security guardrails'', such as posing as a student working on a final university project or writing a paper~\cite{google2025gtigtracker}.

Attackers also wrote such claims and obtained help in an attack with documented victims. In the espionage campaign GTG-1002, which Anthropic attributes with high confidence to a Chinese state-sponsored group, the operators targeted roughly 30 entities, including large technology companies, financial institutions, chemical manufacturers, and government agencies. Anthropic reports how the operators made Claude engage in the attack: ``The key was role-play: the human operators claimed that they were employees of legitimate cybersecurity firms and convinced Claude that it was being used in defensive cybersecurity testing''~\cite[full report, p.~8]{anthropic2025espionage}. AI then performed 80 to 90 percent of the campaign~\cite{anthropic2025espionage}, and Anthropic's investigation validated a handful of successful intrusions in which data was extracted~\cite[full report, pp.~3, 8, 11]{anthropic2025espionage}.

In all three cases, Straiker, Google, and Anthropic each report that a claimed identity or purpose that the attackers placed in the context obtained the models' help.

\begin{keyfinding}
A made-up professional identity or purpose has been enough to bypass deployed safeguards in real-world attacks.
\end{keyfinding}

\section{Prevalence in State-of-the-Art Defenses}\label{sec:census}

Safeguards are scored on text labels (Section~\ref{sec:evaluations}), and deployed safeguards decide from untrusted input that attackers already fabricate (Sections~\ref{sec:obtain} and~\ref{sec:cases}). To measure how far both practices extend in leading research, we collected the defense proposed in each of 129 safeguard papers from the latest editions of 12 top venues, following Arp et al., who examined security papers for evaluation pitfalls~\cite{arp2022dos}. For each defense, we checked whether it reads a credential---an input that attackers cannot change by what they send (Section~\ref{sec:context-def})---as vendors' verification programs do (Table~\ref{tab:mechanisms}); a defense that reads no credential decides only from context (Definition~\ref{def:open-access}). We also recorded how each paper scores safety and whether it discusses attackers who fabricate a professional's identity, purpose, or work setting.

\subsection{Defenses Examined}

The census includes papers that propose a new defense to block requests made with malicious intent without affecting legitimate users (Definition~\ref{def:goal}) and report its safety results, regardless of mechanism and regardless of whether the defense reads context or a credential. Content-category classifiers and capability reduction, such as unlearning, are excluded because they apply one release rule to all requests for the same content and so do not judge who is asking (Section~\ref{sec:mechanisms}). Claude Opus~5.5 read each paper and recorded the inputs to the defense's release decision, with supporting pages or sections, and we reviewed the records.
\subsection{What Defenses Read}

Every defense in the census has the vulnerability, at every venue in Table~\ref{tab:census-venues}, including the security venues USENIX Security, ACM CCS, and NDSS. All but one decide only from context (Definition~\ref{def:open-access}). These 128 defenses decide release from inputs that the user sets or from computations on those inputs, such as activations or a candidate answer (Section~\ref{sec:context-def}). Because each decides on every request that reaches it, dual-use requests included, each decides from untrusted input whether to release dual-use content and therefore has the vulnerability (Section~\ref{sec:obtain}).

Palette~\cite{palette2026}, the sole defense that reads a credential, relaxes refusals in the domains for which a user group is authorized and assumes that this authorization ``is already verified through authentication or access-control mechanisms''. For users without this authorization, its base model decides release from the request text, so the vulnerability remains for them, as it does in vendors' verification programs (Section~\ref{sec:programs}).

\begin{table}[t]
\centering
\caption{Census results by venue.}
\label{tab:census-venues}
\footnotesize
\setlength{\tabcolsep}{4pt}
\begin{tabular}{@{}lrrr@{}}
\toprule
 & & Decide only & Discuss fabricated\tabularnewline
Venue & Papers & from context & claims\tabularnewline
\midrule
ACL & 35 & 35 & 8\tabularnewline
ICLR & 27 & 27 & 3\tabularnewline
ICML & 25 & 25 & 3\tabularnewline
NeurIPS & 15 & 14 & 3\tabularnewline
CVPR & 11 & 11 & 1\tabularnewline
AAAI & 5 & 5 & 0\tabularnewline
ACM CCS & 3 & 3 & 1\tabularnewline
ACM MM & 2 & 2 & 0\tabularnewline
FSE & 2 & 2 & 1\tabularnewline
USENIX Security & 2 & 2 & 1\tabularnewline
NDSS & 1 & 1 & 1\tabularnewline
WWW & 1 & 1 & 0\tabularnewline
\midrule
Total & 129 & 128 & 22\tabularnewline
\bottomrule
\end{tabular}

\smallskip
\parbox{\linewidth}{\raggedright\emph{Discuss fabricated claims}: the paper
discusses attackers who fabricate a professional's identity, purpose, or
work setting.}
\end{table}

\subsection{How the Papers Score Safety}
\label{sec:scoring}

The papers' own safety scores do not detect the vulnerability, because they rest on the text labels of Section~\ref{sec:evaluations}. In 32 of the 129 papers, safety scores compare the defense's predictions with preset safe or unsafe labels; in the other 97, a judge rates the defended model's answers as harmful or as refusals. Labels and ratings alike are attached to text, and for every defense that decides only from context, neither changes with who sends the request. The scores therefore cannot show that a defense distinguishes attackers from professionals who send the same request (Proposition~\ref{cor:scores}).

The three formal safety guarantees in the census rest on the same labels. SELF-REDTEAM~\cite{liu2026chasing} and MAGIC~\cite{wen2026magic} guarantee, once training reaches equilibrium, that their own safety reward model rates the answers safe; BarrierSteer~\cite{tran2026barriersteer} guarantees safety only if its learned safety boundary captures the intended safety property. Each guarantee certifies safety as these labels define it. When the text of a dual-use request does not establish how the answer will be used, its label is wrong for some of the users who send it (Equation~\eqref{eq:label-error}), so on such requests none of the three guarantees shows that the defense keeps the answer from attackers while releasing it to professionals.

\subsection{Whether Papers Discuss Professional Claims}
\label{sec:claims}

Most papers also omit the case through which attackers bypassed deployed safeguards (Section~\ref{sec:cases}). Of the 129 papers, 107 do not discuss attackers who use a professional identity, purpose, or work setting to conceal malicious use (Table~\ref{tab:census-venues}). This is the case that a safeguard exists to handle: a dual-use request whose text does not show how the answer will be used (Section~\ref{sec:mechanisms}), sent by an attacker as legitimate work. Even the 22 papers that do discuss such attackers propose defenses that decide only from context.

Six of these 22 report that requests carrying a professional identity, purpose, or setting written by the user bypass their own defenses. DeepResearchGuard~\cite[p.~43446]{huang2026deep} finds that prompts ``wrapping harmful intent in academic, historical, or analytical registers consistently bypass detection''. ReasoningGuard~\cite[pp.~31525--31526]{wang2026reasoningguard} presents a representative failure: given a request framed as teaching material for instructors, its defended model accepts the stated educational purpose and answers. Reasoning Structure Matters~\cite[pp.~5316--5317]{in2026reasoning} reports that its model can ``fail to recognize harmful intent hidden behind seemingly innocuous phrases like `for educational purposes'\,''.

The other three report the same bypass and trace it to the resemblance between these requests and legitimate work. Into the Gray Zone~\cite[pp.~24836, 24838]{hung2026gray} describes early queries in a conversation as ``indistinguishable from legitimate academic discussions'' and states that its defenses ``do not fully mitigate the vulnerability''. GSAE~\cite[App.~I]{yeon2026graph} traces the requests that evade it to indirect framing such as an ``educational context'', which shifts a request into ``a region that overlaps with legitimate security research''. ShadowMem~\cite[Secs.~5.4.2, 7]{wang2026shadowmem} narrows the requests that bypass it to those ``barely distinguishable from legitimate use'' and states that its judge ``can still be influenced by an adversary that fabricates plausible context''.

Seven of the remaining 16 papers claim that their methods recognize or recover the malicious intent behind a claimed professional purpose. SafeHarbor~\cite[Sec.~3, App.~I, Fig.~13]{liu2026safeharbor} states that it distinguishes legitimate administrative operations from real threats, yet it releases every ambiguous query that its LLM judge reads as a legitimate professional task, and its judging prompt treats the user as an authorized administrator without verifying any identity. Two others check a planned answer against a content policy~\cite{cao2026reasoned} or search for inconsistency between a claimed purpose and the requested information~\cite{yang2026retrieval}. Four more set out to separate malicious intent from the persona or framing around it~\cite{wang2026inde}, \cite[pp.~30133--30134]{huang2026llavashield}, \cite[pp.~41381, 41384]{zhu2026safemt}, \cite{pan2026causal}.

All seven compute this judgment from what the user sends, so by Definition~\ref{def:open-access} an attacker who sends a professional's messages, claimed purpose included, receives the same judgment as the professional. These methods can block requests whose text plainly reveals malice (Definition~\ref{def:plain-malice}), but their claim to recognize malice behind a professional purpose is wrong. Their scores, like every score in the census, rest on labels or ratings attached to text and therefore cannot show that the methods recognize this malice (Proposition~\ref{cor:scores}).

The field's leading defenses thus decide from the same untrusted input as deployed safeguards, and fabricated professional claims bypass both.

\begin{keyfinding}
Every defense in the field's leading research has the vulnerability.
\end{keyfinding}

\section{Why Better Safeguards Cannot Close the Vulnerability}
\label{sec:limits}

The vulnerability lies in the input that a safeguard decides from (Section~\ref{sec:obtain}), yet research and deployments improve how a safeguard judges that input. Safeguards add activation probes~\cite{cunningham2026constitutional}, stacked classifier layers~\cite{sharma2025constitutional}, and intent-aware classifiers~\cite{ferrao2026trueintents}; CarryOnBench asks models to revise their judgment of intent when users clarify it in later turns~\cite{zheng2026carryon}. No such improvement can close the vulnerability, because each adds training data or computation to a safeguard that still decides only from context.

More training data cannot close it: a label on the text of a dual-use request that does not plainly reveal malice is wrong for some of the users who send it (Equation~\eqref{eq:label-error}), and even with labels correct for every occurrence, a safeguard that decides only from context cannot tell attackers from professionals who send the same request (Proposition~\ref{cor:scores}). More computation cannot close it either: whatever a safeguard computes from the messages it receives, an attacker who sends a professional's messages receives each answer as often as the professional does (Section~\ref{sec:decide-only}).

\subsection{Any Computation on Context}
\label{sec:decide-only}

Definition~\ref{def:open-access} limits what a safeguard reads but leaves open what it computes. The model below covers every computation that a safeguard deciding only from context can run, and therefore every possible improvement to its judgment.

Fix a task \(w\). Over \(T\) turns the user sends a message \(x_t\), and the
safeguard \(D\), together with the model it guards, returns a reply \(y_t\), such
as a question about who is asking. The \emph{exchange}
\(h=(x_1,y_1,\ldots,x_T,y_T)\) records all messages and replies;
\(h_{t-1}\) denotes its first \(t-1\) turns. After the last turn, \(D\) produces an answer
\(a\) from a set \(\mathcal A\) that contains the refusal. The context is \((w,h)\), and \(H\) and \(A\) denote the
exchange and the answer as random variables.

The safeguard may compute anything from the context. Its internal state \(\xi\),
which the user never sees, holds everything it computes: activations, candidate answers,
classifier flags, and the outputs of retrieval and tool calls. At each turn \(D\)
draws its internal state and reply, and after the last turn its answer, with
probabilities that depend on the task, the exchange so far including the current
message, and the earlier internal state---and on nothing else about the user. Let \(\Lambda_D(h,\xi,a)\) be the product of these probabilities along
\((h,\xi,a)\). Text that the user sees before an output classifier withdraws it is
part of the reply.

Each user chooses its own messages. At turn \(t\) a user chooses \(x_t\) with
probability \(m_t(x_t\mid w,h_{t-1})\) from a single set of messages that is the same
for every user: any message that one user can send, another can send as well
(Section~\ref{sec:threat-model}). What the user knows, such as its role and
purpose, reaches the safeguard only through the messages the user chooses. For a user who
chooses by \(m=(m_1,\ldots,m_T)\), the exchange, internal state, and answer therefore have the distribution
\begin{equation}
P_m^D(h,\xi,a)=\textstyle\prod_{t=1}^{T}m_t(x_t\mid w,h_{t-1})\,\Lambda_D(h,\xi,a).
\label{eq:internal-state-law}
\end{equation}
In Equation~\eqref{eq:internal-state-law} the user appears only through \(m\), and the factor \(\Lambda_D\) is the same for every user, as Definition~\ref{def:open-access} requires.

Users ask for a deliverable in many ways; let \(m\) be any one of them. A professional and an attacker who both choose their messages by \(m\) share the distribution \(P_m^D\), so both receive an answer in any set \(\mathcal K\subseteq\mathcal A\) with the same probability \(P_m^D(A\in\mathcal K)\). Thus, whatever way of asking a professional uses and however the safeguard is designed, an attacker who asks the same way receives each answer as often as the professional does (Equation~\eqref{eq:single}).

A user \emph{produces} an exchange \(h\) if \(P_m^D(h)>0\). The factors \(m_t\) multiply both \(P_m^D(h,\xi,a)\) and \(P_m^D(h)\), so they cancel when the first is divided by the second. Given \(h\), the internal state and answer therefore follow the distribution \(\Lambda_D(\cdot\mid h)\), which is \(\Lambda_D(h,\cdot,\cdot)\) normalized over \((\xi,a)\), for every user who produces \(h\). Let \(p\) be the probability of a refusal under \(\Lambda_D(\cdot\mid h)\). Proposition~\ref{cor:scores} then holds with the request \(x\) replaced by any exchange that both attackers and professionals produce: however a safeguard that decides only from context is trained, it cannot distinguish the attackers who produce such an exchange from the professionals who produce it.

A user's later messages depend on the safeguard's earlier replies. Lemma~\ref{lem:turn} (Appendix) shows that the cancellation still holds at every turn: given the exchange so far and the current message, the internal state and the next reply have the same distribution for every user.

\subsection{Improvements to the Judgment}
\label{sec:block-rate}

Activation probes, output classifiers, and the intent recognition that seven census papers claim (Section~\ref{sec:claims}) all compute from context and internal state alone, so a safeguard that adds them still treats every user who sends the same messages alike (Equation~\eqref{eq:internal-state-law}). The same holds for stacked layers, follow-up questions, and linked sessions.

\paragraph{Stacked layers}
A stack of classifiers is one safeguard. Its internal state holds the candidate
answer and the flags \(F_1,\ldots,F_k\) of its layers; it releases the
candidate when no layer raises a flag. Given an exchange \(h\), the flags follow the distribution \(\Lambda_D(\cdot\mid h)\), so every user who produces \(h\), attacker or professional, passes the stack with the same probability
\[
1-\Lambda_D\Bigl(\textstyle\bigcup_{i=1}^k\{F_i=1\}\Bigm| h\Bigr)
\leq 1-\max_i\Lambda_D(F_i=1\mid h).
\]
A stack passes an exchange at most as often as its strictest layer does, and each
added layer can only lower this chance---by the same amount for the attacker and
the professional who produce the exchange.

\paragraph{Follow-up questions}
Follow-up questions are replies of the safeguard. By Lemma~\ref{lem:turn}, which question it asks next depends on the exchange so far, not on the user. The questions distinguish attackers from professionals only when an attacker's answers plainly reveal malice. A question about the user, such as its authorization or purpose, is answered by a claim, and the safeguard treats that claim the same whether it is true or false. To any other question, an attacker can send any reply that a professional can send (Section~\ref{sec:threat-model}), and the safeguard treats every user who sends the same replies alike.

\paragraph{Linked sessions}
\label{sec:splitting}
A safeguard can also judge a request together with the user's earlier sessions, linked for example by account. The linked history consists of the user's own earlier messages, so it is context; every signal computed from it is also context (Section~\ref{sec:context-def}), and the linked sessions form one exchange. When the linked messages together still do not plainly reveal malice, the safeguard still gives attackers whatever it gives professionals who send the same messages. Records about the account that attackers cannot change by what they send, such as its age or payment method, are credentials (Definition~\ref{def:credential}), and attackers receive content released on such a record at most as often as they obtain such a record (Proposition~\ref{thm:discount}).

\subsection{Four Choices}
\label{sec:four-choices}

All of these safeguards decide only from context, so on a dual-use request that does not plainly reveal malice---such as the requests that the attackers in Section~\ref{sec:cases} sent---they all face the same choices. The text does not establish how the answer will be used, and any user can write it. By Definition~\ref{def:open-access} the safeguard responds alike to every user who sends the request, so it has four choices.
\begin{enumerate}
\item \emph{Block the request.} The attacker asks for the same work in other words. Blocking every wording withholds the work from everyone (choice~4). Any wording that the safeguard allows can also be used by attackers (choices~2 and~3).
\item \emph{Release it.} The attacker receives the answer.
\item \emph{Block it at random.} Whoever sends the request has the same chance of receiving the answer. An attacker who is refused can rephrase and try again. Rephrasing after a refusal does not itself plainly reveal malice, since OpenAI tells flagged users to ``try rephrasing your request''~\cite{openai2026gpt55cyber}.
\item \emph{Withhold the content from everyone.} Neither the attacker nor anyone else receives it. This is capability reduction, the work of content-category classifiers.
\end{enumerate}

Choices~1 through~3 give an attacker whatever they give anyone who sends the request, so the safeguard keeps an answer from attackers only by withholding it from everyone. On such a request it therefore cannot meet its goal of blocking malicious requests without affecting legitimate users (Definition~\ref{def:goal}). A better safeguard catches more requests that plainly reveal malice but faces the same four choices on every dual-use request that does not. A component that decides only from context can therefore secure two decisions against what attackers write: blocking a request whose text plainly reveals malice, for every user who sends it, and applying one release rule to all requests for the same content, as a content-category classifier does.

\section{Patch: Each Component Secures What It Can}
\label{sec:credentials}

The two decisions that a component deciding only from context can secure (Section~\ref{sec:four-choices}) are the decisions that the text settles: whether a request plainly reveals malice, and what the content of an answer is. A claimed identity or purpose changes neither, so a label on either is correct for every user (Sections~\ref{sec:mechanisms} and~\ref{sec:scores}). The text does not settle who should receive dangerous dual-use content, so this decision needs an input that no claim can change: a credential that the user holds. Deployments already run a component for each of these three decisions: safeguards, content-category classifiers, and, for some users, credentials. The patch adds no component; it gives each existing component only the decisions it can secure (Table~\ref{tab:roles}).

\begin{table}[t]
\centering
\caption{Roles of the three existing components today and under the patch.}
\label{tab:roles}
\footnotesize
\setlength{\tabcolsep}{3pt}
\renewcommand{\arraystretch}{1.1}
\rowcolors{2}{white}{zebra}
\begin{tabular}{p{\dimexpr0.25\linewidth-6pt\relax}p{\dimexpr0.35\linewidth-6pt\relax}p{\dimexpr0.40\linewidth-6pt\relax}}
\toprule
\raggedright Component & \raggedright Today & \raggedright Under the patch\tabularnewline
\midrule
\raggedright Content-category classifier & \raggedright Decides release of some categories & \raggedright Applies each category's rule; marks dangerous dual-use content\tabularnewline
\raggedright Safeguard & \raggedright Decides dual-use release by judging use & \raggedright Blocks requests plainly revealing malice\tabularnewline
\raggedright Credential & \raggedright Relaxes safeguards or admits narrow groups & \raggedright Decides who receives dangerous dual-use content\tabularnewline
\bottomrule
\end{tabular}
\end{table}

\subsection{What Each Component Secures}
\label{sec:roles}

The safeguard blocks requests that plainly reveal malice (Definition~\ref{def:plain-malice}). It can secure this decision because the text settles it and the safeguard treats all users who produce the same exchange alike (Definition~\ref{def:open-access}). Blocking a message whose text shows harmful use or an evasion attempt therefore blocks every user who sends it, whatever identity or purpose the message also claims. The safeguard makes no other release decision, because on any other dual-use request it can keep an answer from attackers only by withholding it from everyone (Section~\ref{sec:four-choices}).

The content-category classifier assigns each answer's content to a category and applies that category's release rule. It can secure these rules because a claimed identity or purpose does not change what the content is, and therefore does not change its category (Section~\ref{sec:mechanisms}). Prohibited content is withheld from everyone. The deployer chooses which dual-use categories it will not release to every user, such as the high-risk dual-use tiers of Section~\ref{sec:default-access}; content in these categories is \emph{dangerous dual-use content}, and content that is neither prohibited nor dangerous is released to everyone. Because a rule that treats every user alike either blocks professionals from dangerous dual-use content or serves attackers (Section~\ref{sec:four-choices}), the classifier only marks this content and leaves its release to the credential.

The credential decides who receives dangerous dual-use content. It can secure this decision because attackers cannot change a credential by what they send. A rule on context cannot: by Equation~\eqref{eq:single}, an attacker can obtain this content whenever a rule on context releases it to anyone. A credential may record a verified role or authorization, a verified project, or a binding to a verified real-world identity. An example is the tier that Anthropic's Cyber Verification Program (CVP) assigns after verifying an applicant's identity and organization~\cite{anthropic2026cvp,anthropic2026cvphelp}. What a user writes is not a credential: the CARLIY0S operators wrote that their work was an authorized red-team engagement, and that description is context.

\subsection{Release by Credential}
\label{sec:credential-def}

To state what the credential secures, fix a task \(w\). An input \(S\), checked at release, has distribution \(Q^S\) among the attackers who ask for \(w\), whether all attackers or any single group of them. The distribution \(Q^S\) gives how often attackers present each value at release, by any route: applying, forging a credential, receiving a transferred one, compromising an account, and holders' own misuse. We say that attackers \emph{obtain} a value when they present it at release.

\begin{definition}[Credential]\label{def:credential}
The input \(S\) is a \emph{credential} if attackers cannot change its
distribution by the messages they send to the deployment: whatever messages
attackers send it, they present \(S\) with distribution \(Q^S\).
\end{definition}

\label{sec:discount}\label{sec:discount-theorem}
A deployment that reads \(S\) fits the model of Section~\ref{sec:decide-only} with
one additional input: the user's messages and the deployment's draws may also depend on
the value \(s\) that the user holds.
Equation~\eqref{eq:internal-state-law} then holds among users who hold
the same value, so among them the deployment decides only from context: its safeguard treats all holders who produce the same exchange alike and can still block holders whose requests plainly reveal malice.

\begin{proposition}[Release under the patch]
\label{thm:discount}\label{thm:discount-informal}
Let a deployment \(D\) read a credential \(S\) and the context, and let
\(\mathcal K\) be a set of answers, such as the answers that a content-category
classifier marks as dangerous dual-use content. Let \(D\) give answers in
\(\mathcal K\) only to users who present a value in a set
\(\mathcal S_{\mathcal K}\). Whatever messages attackers send,
\[
\Pr[\text{an attacker receives an answer in }\mathcal K]
\leq Q^S(\mathcal S_{\mathcal K}).
\]
\end{proposition}

\begin{proof}
An attacker who presents a value outside \(\mathcal S_{\mathcal K}\) receives no
answer in \(\mathcal K\), so the probability is at most the probability
\(Q^S(\mathcal S_{\mathcal K})\) that an attacker presents a value in
\(\mathcal S_{\mathcal K}\). By Definition~\ref{def:credential}, this probability
does not change with what attackers send.
\end{proof}

When \(\mathcal K\) contains the answers that a content-category classifier marks as dangerous dual-use content, attackers therefore receive such an answer at most as often as they obtain a credential for it, whatever they write. A deployment that releases this content without a credential sets no such limit: attackers can obtain any of this content that any user receives (Equation~\eqref{eq:single}).

\label{sec:limitations}\label{sec:issuance}
For the credential to guarantee that attackers rarely receive this content, they must rarely obtain the values in \(\mathcal S_{\mathcal K}\). The deployment chooses which values enter \(\mathcal S_{\mathcal K}\), and the issuer's checks decide how often attackers obtain these values by applying. Attackers obtain a value less often than legitimate users hold it when the issuer verifies a fact that is true of professionals and rarely true of attackers. They also obtain it less often when holders face consequences for misuse, as they do when the value is bound to a verified real-world identity.

An issuer that reads only what applicants state verifies no such fact. Such an issuer itself decides only from context: applicants are its users, the application is their context, and the value it issues is its answer. By Equation~\eqref{eq:single}, attackers who submit an application that legitimate applicants submit obtain each value as often as those applicants do, however the issuer judges the application. Issuing a credential only on what applicants state therefore leaves the vulnerability open.

\subsection{Verification Programs and Deployment}
\label{sec:programs}\label{sec:verification-exemptions}

Vendors' verification programs already issue credentials, and these credentials already change release decisions. In a controlled test on Claude Opus~5.5, Anthropic ran ten offensive scenarios five times each~\cite{anthropic2026cvp}. At the Defense tier, 46 of the 50 trials were blocked at some point; at the Red Team tier, none were. No request changes the tier, and release followed it. Anthropic has also released an entire model by credential: it first launched its frontier model Mythos only to selected, vetted organizations because of the model's advanced capabilities in cyber vulnerability discovery~\cite[pp.~51--52]{anthropic2026s1}.

Vendors' programs use credentials in two ways. First, a credential relaxes a safeguard that still judges use. Vetted users with dual-use applications may receive answers that some ASL-3 classifier actions would otherwise block~\cite{anthropic2025asl3}. OpenAI's Trusted Access for Cyber gives verified users ``lower classifier-based refusals''~\cite{openai2026gpt55cyber}, and safeguards remain at the CVP's Defense tier~\cite{anthropic2026cvp}. Second, a credential admits only a narrow group, such as the CVP's Red Team tier~\cite{anthropic2026cvp} or the LSVP's High-risk Use grant for Mythos, which is limited to a small set of entities with additional vetting~\cite{anthropic2026lsvp}.

In both cases, the vulnerability persists for every user without the credential: what the user writes can still decide whether it receives high-risk dual-use content (Section~\ref{sec:default-access}). Among holders, a safeguard that still judges use decides only from context (Section~\ref{sec:credential-def}), so it can withhold an answer from an attacker who holds the credential only by withholding it from honest holders who produce the same exchange. Under the patch, the credential decides for every user whether it receives dangerous dual-use content: users without the credential do not receive this content, whatever they write, and among holders the safeguard blocks only requests that plainly reveal malice.

\label{sec:deploying}
The patch also changes what professionals receive. At general access, safeguards block work that professionals need, while attackers who wrote a professional purpose obtained such help (Section~\ref{sec:cases}). Fable~5.1's safeguards are designed to keep blocking vulnerability discovery in compiled binaries and the high-risk dual-use cyber tier~\cite[p.~52]{anthropic2026fable51card}. Anthropic describes the activities in this tier as ``part of the day-to-day work of cybersecurity professionals''~\cite{anthropic2026safeguardsframework}, and recent Anthropic models ``restrict access to a wide range of dual-use biological research queries''~\cite{anthropic2026threatsep}. Under the patch, every professional who holds a credential for dangerous dual-use content receives help with this work, as Anthropic's programs already offer to some verified professionals~\cite[p.~46]{anthropic2026fable51card},~\cite{anthropic2026lsvp}. Users without a credential still receive content that is not dangerous.

\begin{takeaway}
Under the patch, attackers can no longer obtain dangerous dual-use content by what they write.
\end{takeaway}

\section{Related Work}\label{sec:related}

\paragraph{Dual use}
Dual-use research studies how the same knowledge and technologies serve both beneficial and harmful purposes and how access to them can be governed~\cite{nrc2004biotechnology,miller2007dualuse,forge2010dualuse}. In AI, studies document the misuse potential of models built for beneficial work~\cite{urbina2022dualuse} and propose governing large language models as dual-use research of concern~\cite{grinbaum2024dualuse}; others argue that who may access which features of a model should be a subject of AI governance~\cite{kembery2024modelaccess}. Proposed controls on dangerous dual-use capabilities include access through an API~\cite{urbina2022dualuse}, structured access for approved users~\cite{li2024wmdp}, and access controls that pair content categories with user verification~\cite{wybitul2025access}. Wybitul states the weakness itself, but only as an observation motivating access control: the context that current safety systems infer comes from user-supplied input that attackers can fabricate~\cite{wybitul2025access}. We establish it as a vulnerability: attackers exploit it in real attacks on deployed models (Section~\ref{sec:cases}), every defense in the field's leading research has it (Section~\ref{sec:census}), and no improvement closes it in a safeguard that decides only from context (Section~\ref{sec:limits}). These results also identify which decision access control must take over: who receives dangerous dual-use content. Safeguards and content-category classifiers keep the decisions they can secure (Section~\ref{sec:credentials}).

\paragraph{Misuse safeguards and their evaluation}
Misuse safeguards block harmful requests through safety training, text and activation classifiers, and reasoning over safety policies~\cite{guan2024deliberative,inan2023llamaguard,han2024wildguard,sharma2025constitutional,cunningham2026constitutional}. Systematic studies evaluate guardrails for security, efficiency, and utility~\cite{wang2025guardsok,xin2026filters}, scoring safety against labels attached to request text or against a judge's ratings of answers. Wu et al.\ show that such scores describe a safeguard only on the requests it was tested on, so a higher score does not by itself establish that a deployed system gives attackers less help~\cite{wu2026safeguardworked}. We show that a text label on a dual-use request that does not plainly reveal malice is wrong for some of the users who send it, and that the attack success rate on such a request equals the rate at which professionals receive the answer. Neither accuracy against such labels nor attack success rates can therefore show that a safeguard distinguishes malicious from legitimate use (Section~\ref{sec:evaluations}).

Other work establishes what safeguards cannot do: stronger adaptive attacks bypass published defenses~\cite{nasr2025attacker}, and proofs establish limits for efficient filters outside the model~\cite{ball2025impossibility}, continuous utility-preserving wrappers~\cite{bhatt2026defensetrilemma}, semantic censorship of outputs~\cite{glukhov2024position}, and defenses that limit what answers leak~\cite{glukhov2024leaks}. These limits hold for particular designs. Ours follows from who writes the input: attackers can obtain whatever a safeguard that decides only from context releases to professionals, however the safeguard is built and however much it computes, including when it reads the model's activations (Sections~\ref{sec:obtain} and~\ref{sec:limits}).

\paragraph{Untrusted input}
Security engineering treats input that an attacker can modify as untrusted. CWE-807 names the weakness of a protection mechanism that relies on such input~\cite{mitre2026cwe807}, and NIST's identity guidelines treat attributes that a subject provides without verification as self-asserted~\cite{nist2025sp80063a4}. Decisions that must resist attackers rest instead on verified evidence, such as platform attestation~\cite{parno2010bootstrapping} and anonymous or personhood credentials~\cite{davidson2018privacypass,shih2025zkpromises,adler2024personhood}. These works address who a user is, what it may do, or what a system runs---not how a model's answer will be used. We show that misuse safeguards judge how an answer will be used from untrusted input: the use reaches them only through what the user writes (Section~\ref{sec:obtain}).

\section{Conclusion}\label{sec:conclusion}

Misuse safeguards for LLMs commit a basic security error: they decide whether to release dual-use content from what users write, including the professional identity or purpose they claim. A security decision that rests on attacker-controlled input is a weakness that security engineering already catalogues (CWE-807). Attackers who fabricate a professional identity or purpose have already obtained help from deployed models for espionage and for attacks on banks.

The field's leading research makes this error in every defense it proposes, including those published at security venues. It also studies the wrong object: it treats harm as a property of the text of a request, although the same text can come from a professional or an attacker. Its benchmarks therefore cannot reveal the vulnerability: a safeguard can score high on them and still release dangerous dual-use content to attackers (Proposition~\ref{cor:scores}). Nor can these scores support the claim of several of its papers that their defenses recognize the malice behind a claimed professional purpose.

Better safeguards cannot close the vulnerability. On a dual-use request that does not plainly reveal malice, a safeguard that decides only from what users write keeps the answer from attackers only by withholding it from everyone. The patch gives each existing component only the decisions it can secure: safeguards block requests that plainly reveal malice, content-category classifiers apply each category's release rule and mark dangerous dual-use content, and a credential decides who receives dangerous dual-use content. Attackers then no longer obtain this content by writing a claim, and professionals who hold a credential receive help that is now blocked at general access.

\section*{Responsible Disclosure}

Attackers already exploit this vulnerability on deployed commercial models (Section~\ref{sec:cases}). We disclosed it to model vendors; OpenAI, Anthropic, Google, and the Qwen team at Alibaba have responded. The work involved no human subjects.

\section*{Acknowledgment}
The authors used GPT-6 Luna (OpenAI) to screen titles and abstracts for the census, Claude Opus~5.5 (Anthropic) to read the papers in the census (Section~\ref{sec:census}) and to translate and edit the manuscript, and Claude Opus~4.6 (Anthropic) to polish the language. The authors reviewed all outputs and take full responsibility for the content.

\bibliography{references}

\appendix[The Model at Each Turn]
A user's later messages depend on the safeguard's earlier replies, which in turn depend on its internal state. The lemma shows that, given the exchange so far and the current message, what the safeguard computes at each turn does not depend on the user.

All sets are finite. At turn \(t\), \(D\) draws its internal state \(\xi_t\) and
its reply \(y_t\) from a probability kernel
\(K_t(\xi_t,y_t\mid w,h_{t-1},x_t,\xi_{<t})\), and after turn \(T\) it draws
\(\xi_{T+1}\) and the answer \(a\) from \(K_{T+1}(\xi_{T+1},a\mid w,h,\xi_{\leq T})\).
With \(\xi=(\xi_1,\ldots,\xi_{T+1})\), \(\Lambda_D(h,\xi,a)\) is the product of
\(K_1,\ldots,K_{T+1}\) along \((h,\xi,a)\).

\begin{lemma}[Each turn]
\label{lem:turn}
Let \(D\) decide only from context. At every turn \(t\), given the exchange so far \(h_{t-1}\) and the current message \(x_t\), the internal state \(\xi_{\leq t}\) and the reply \(y_t\) have the same distribution for every user for whom \((h_{t-1},x_t)\) has positive probability.
\end{lemma}

\begin{proof}
Sum Equation~\eqref{eq:internal-state-law} over everything after the reply
\(y_t\), working backward from the last factor. Each factor sums to one over the variables it draws, so
every later factor drops out. The resulting distribution of
\((h_{t-1},x_t,\xi_{\leq t},y_t)\) is
\(\prod_{\tau\leq t}m_\tau(x_\tau\mid w,h_{\tau-1})
\prod_{\tau\leq t}K_\tau(\xi_\tau,y_\tau\mid w,h_{\tau-1},x_\tau,\xi_{<\tau})\).
The message factors depend only on \((h_{t-1},x_t)\), so they cancel when conditioning on \((h_{t-1},x_t)\). The conditional distribution of \((\xi_{\leq t},y_t)\) is therefore determined by the kernels of \(D\) alone.
\end{proof}

\end{document}